\documentclass[runningheads]{llncs}
\usepackage[T1]{fontenc}
\usepackage{graphicx}
\usepackage{amsmath,amssymb,amsfonts}
\usepackage{hyperref}
\usepackage{url}
\usepackage{xcolor}
\usepackage{booktabs}
\usepackage{microtype}

\newcommand{\WPI}{\Phi}
\newcommand{\kBT}{k_B T}
\newcommand{\Erasure}{c}
\newcommand{\IntOp}{I}
\newcommand{\Sys}{\Sigma}
\newcommand{\Nop}{N}
\newcommand{\Fover}{F}
\newcommand{\Alg}{\mathcal{A}}
\newcommand{\Hw}{\mathcal{H}}
\newcommand{\Task}{\mathcal{T}}
\newcommand{\Class}{\mathcal{C}}
\newcommand{\Barrier}{\mathcal{B}}
\newcommand{\descop}{\mathrm{desc}}
\newcommand{\idle}{\mathrm{idle}}
\newcommand{\Erestore}{E_{\mathrm{restore}}^{\mathrm{cycle}}}
\newcommand{\Eadapt}{E_{\mathrm{adapt}}}
\newcommand{\poly}{\mathrm{poly}}
\newcommand{\CL}{\mathsf{CL}}
\newcommand{\ZPP}{\mathsf{ZPP}}
\newcommand{\TCone}{\mathsf{TC}^1}
\newcommand{\MutInfo}{I_{\mathrm{alg}}}
\newcommand{\Herase}{H_{\mathrm{erase}}}

\begin{document}

\title{Watts-per-Intelligence Part II:\\Algorithmic Catalysis}
\titlerunning{Watts-per-Intelligence II: Algorithmic Catalysis}
\author{Elija Perrier\inst{1}}
\authorrunning{E. Perrier}
\institute{Centre for Quantum Software and Information, University of
Technology Sydney, Australia.\\ \email{elija.perrier@gmail.com}}

\maketitle

\begin{abstract}
We develop a thermodynamic theory of algorithmic catalysis within the watts-per-intelligence framework, identifying reusable computational structures that reduce irreversible operations for a task class while satisfying bounded restoration and structural selectivity constraints. We prove that any class-specific speed-up is upper-bounded by the algorithmic mutual information between the substrate and the class descriptor, and that encoding this information incurs a minimum thermodynamic cost via Landauer erasure. Combining these results yields a coupling theorem that lower-bounds the deployment horizon required for an algorithmic catalyst to be energetically favourable. The framework is illustrated on an affine-SAT class and situates contemporary learned systems within an information–thermodynamic constraint on intelligent computation.
\keywords{Intelligence \and Catalytic Computing \and Kolmogorov Complexity \and Thermodynamics
of Computation \and Algorithmic Thermodynamics}
\end{abstract}

\section{Introduction}
\label{sec:intro}
Catalysts are central to the chemical processes that drive organic life by permitting reactions that are thermodynamically allowed, but kinetically infeasible \cite{laidler1987chemical}. In this sense, catalytic structure is one of the enabling conditions
for complex organic life and intelligence. These structures, which include enzymes, ribozymes, metal centres and mineral surfaces, make
reaction pathways available that would otherwise be unusable on
biological timescales. A transformation may be thermodynamically
permitted, in the narrow sense that its products are not forbidden by
equilibrium free energy, yet still remain practically inaccessible
because the relevant transition state is too unlikely. Catalytic
structure may make such transformations accessible via lowering activation free energy,
stabilising transition-state structure, pre-organising substrates,
coupling unfavourable steps to energetic cofactors and suppressing
competing routes~\cite{wolfenden2001depth,benkovic2003enzyme,garcia2010enzyme}. Recent work on machine intelligence \cite{ebtekar2025foundations,perrier2025watts,ebtekar2021information} has shown how thermodynamics constrains algorithmic forms of intelligence. It is therefore natural to ask whether there exist algorithmic analogues of
catalysts which may help overcome such constraints: reusable computational structures that render feasible certain algorithmically intelligent activities that would otherwise be difficult to implement. 

We answer this question in the affirmative by extending and developing algorithmic analogues of catalytic structure characterised by three properties: (a) affording a lower-energy pathway through a structured family of transitions; (b) remaining persistent so as not to be consumed by the reactions they enable; and (c) imposing a structural selection mechanism that facilitates algorithmic computation by virtue of their structure. To do so, we draw upon and synthesise elements of two primary existing bodies of work: (i) that of catalytic computation~\cite{buhrman2014computing} in which a small working machine is given access
to auxiliary memory whose contents must be restored exactly at
termination; and (ii) the thermodynamics of intelligence, specifically the watts-per-intelligence framework \cite{perrier2025watts}. Our novel contribution involves the introduction of a thermodynamic basis for the catalytic effects of computation connecting to emerging work on the thermodynamics of machine intelligence. Specifically, we contribute the following: 
\begin{enumerate}
    \item A \emph{selectivity theorem} such that in a
universal-search cost model the logarithmic class-specific speed-up
achievable by a substrate is bounded above by the algorithmic mutual
information between the substrate description and a canonical
descriptor of the task class.
\item A \emph{physical-erasure lemma},
obtained from the Zurek--Bennett correspondence between algorithmic
complexity and thermodynamic
entropy \cite{zurek1989algorithmic,bennett1982thermodynamics},
which sets out a lower bound on the
number of irreversible bit-operations a physical adaptation
process must perform on the adaptation substrate to encode information not otherwise present.
\item A \emph{coupling theorem} which proposes a lower bound on the energetic break-even horizon below which no
putative catalyst can be thermodynamically favourable on its task class.
\item A \emph{composition theorem}
for chains of catalysts which establishes that catalytic gains multiply in
rate while combining sub-additively in information, matching the
chemical intuition that enzymes along metabolic pathways compound
their rate enhancements without compounding their selectivity.
\end{enumerate}
   We demonstrate the application of these results using a worked affine-SAT example and show how the catalytic-computation model of Buhrman et~al. \cite{buhrman2014computing} is recovered using our formalism. First, we recapitulate the WPI framework for thermodynamic constraints on intelligence.

\section{Background and Related Work}
\label{sec:background}

\subsection{Watts-per-intelligence}

The watts-per-intelligence framework sets out a basis for estimating thermodynamic bounds on certain algorithmic computation associated with intelligent tasks \cite{perrier2025watts}. A computational system is represented as a triple
$\Sys=(\Alg,\Hw_{\mathrm{exec}},\Hw_{\mathrm{adapt}})$, where $\Alg$
is the algorithm, $\Hw_{\mathrm{exec}}$ is the substrate on which
deployment occurs, and $\Hw_{\mathrm{adapt}}$ is the substrate on
which training, compilation, restoration or repair occurs. The two
substrates may coincide when adaptation and deployment are performed
on the same physical device. The distinction is informative because substrate overhead factors and energy accounting may differ
substantially between the two phases. Let $k_B$ be Boltzmann's
constant, let $T$ be the ambient temperature of the computation, and
let $\Erasure:=\kBT\ln 2$ denote the Landauer cost per erased bit. Let $\Nop(\Sys,\Task)$ count the irreversible bit-operations performed
on $\Hw_{\mathrm{exec}}$ during a benchmark pass of horizon $\tau$ on
task $\Task$. Reconfiguration erasures on $\Hw_{\mathrm{adapt}}$ are
accounted separately. Each substrate carries an overhead factor
$\Fover(\Hw)\ge 1$ relative to the Landauer benchmark~\cite{perrier2025watts,frank2005introduction,athas1994low}. Applying the second law to irreversible operations on the relevant
substrate gives the substrate-level energy floor:
\begin{equation}
\label{eq:energy-floor}
E(\Sys,\Task)\;\ge\;\Fover(\Hw_{\mathrm{exec}})\,\Nop(\Sys,\Task)\,
\Erasure.
\end{equation}
A task-weighted intelligence score
$\IntOp(\Sys,\Task)=\sum_i w_i p_i(\Sys)$ with $w_i\ge 0$ and
$p_i\in[0,1]$ is fixed, corresponding to a finite-task restriction
of Legg--Hutter universal intelligence~\cite{legg_universal_2007,hutter2004universal}. The watts-per-intelligence ratio and its
Landauer floor are then
\begin{equation}
\label{eq:wpi}
\WPI(\Sys,\Task)=\frac{E(\Sys,\Task)/\tau}{\IntOp(\Sys,\Task)},
\qquad
\WPI^{\downarrow}(\Sys,\Task)
=\frac{\Fover(\Hw_{\mathrm{exec}})\Nop(\Sys,\Task)\Erasure}
{\tau\IntOp(\Sys,\Task)},
\end{equation}
with $\WPI^\downarrow\le\WPI$. 
Equation 
\eqref{eq:energy-floor} provides the basis for calculating the thermodynamic costs against which our proposed catalytic structures are compared.

\subsection{Catalytic computation}
In catalytic computation literature, a catalytic Turing machine in the sense of Buhrman et~al.~\cite{buhrman2014computing}
is a logspace machine equipped with an additional read-write
auxiliary tape of length polynomial in the input size, whose initial
contents are arbitrary and whose final contents must coincide
bit-exactly with the initial contents at termination. The class $\CL$
is defined as the set of languages decidable by such a machine in
polynomial time with logarithmic clean space. The central
structural result from ~\cite{buhrman2014computing,koucky2016catalytic} of interest to us is that log-space uniform $\TCone$ is contained in
$\CL$ and that $\CL$ is in turn contained in
$\ZPP$, with
subsequent refinements in the randomised, non-deterministic and
derandomised settings \cite{buhrman2018catalytic,girard2013catalytic,pyne2023catalytic}
and a tree-evaluation algorithm in space
$O(\log n\cdot\log\log n)$~\cite{cook2024tree,cook2025lower,goldreich2024lower}. 
This catalytic computational result shows that a very limited machine (with only log-space memory) can solve much richer problems if it is allowed to use a large auxiliary workspace, provided that workspace is returned exactly to its original state at the end. This use without consumption is what makes the model mathematically clean: the auxiliary tape can start in any configuration, so the computation must work uniformly without relying on hidden structure. However, the model is purely asymptotic and does not account for the physical cost of restoring that workspace. This is precisely the gap addressed by our generalised algorithmic catalysis framework in Section~\ref{sec:framework}.

\subsection{Related work}
While several literatures address parts of this problem, none address the question of thermodynamically-based algorithmic catalysis per se. \emph{Amortised inference} in probabilistic modelling~\cite{stuhlmuller2013amortized,kingma2013autoencoding,rezende2014stochastic} replaces repeated expensive inference with a learned procedure that is cheap to apply at deployment, which motivates the adaptation/deployment split but does not track the underlying thermodynamic cost. \emph{Algorithmic thermodynamics}, beginning with Zurek~\cite{zurek1989algorithmic} and Bennett~\cite{bennett1982thermodynamics,bennett1998information} and extended in stochastic thermodynamics~\cite{parrondo2015thermodynamics,wolpert2019stochastic}, links algorithmic complexity to physical entropy and provides the basis for bounding erasures in terms of information, as used in Section~\ref{sec:coupling}. \emph{Minimum description length} and related learning bounds~\cite{rissanen1978modeling,grunwald2007minimum,mcallester1999pacbayes} quantify how much information a model captures about a data-generating process, which parallels the structural selectivity condition here, albeit phrased in terms of distributions rather than class descriptors. Finally, work on the thermodynamics of learning~\cite{still2012thermodynamics,goldt2017stochastic} connects training-time dissipation to mutual information between learner and data in concrete physical models, providing a coarse-grained analogue of the adaptation-cost lemma in Section~\ref{sec:coupling}. 

\section{Formal Framework}
\label{sec:framework}
We now set up our formal framework. Let $U$
denote a universal prefix Turing machine. For any string or structure
$x$ admitting a canonical encoding, $K_U(x)$ denotes prefix
Kolmogorov complexity on $U$ and $K_U(x\mid y)$ denotes its
conditional counterpart. The algorithmic mutual information between
$x$ and $y$ is given by:
\begin{equation}
\label{eq:mutinfo}
\MutInfo(x:y):=K_U(x)-K_U(x\mid y),
\end{equation}
which is symmetric up to a small additive $O(\log K)$ term \cite{li2019kolmogorov,kolmogorov1965three,solomonoff1964formal,levin1984randomness}.

For a substrate $\Hw$, $\descop(\Hw)$ denotes a fixed,
self-delimiting encoding of $\Hw$ on the universal machine $U$,
i.e.\ a canonical description that fully specifies the substrate
independently of the task or its operational history. For an algorithm $\Alg$, $\idle(\Alg)$ denotes a fixed
reference (idle) state to which the system is required to return after
each cycle. Fixing it in advance prevents arbitrary final states
from being retroactively declared idle and thereby trivialising
the bounded-restoration condition. For a task class
$\Class$, $\sigma(\Class)$ denotes a canonical descriptor of the
generative or structural regularities of $\Class$: the shortest description of the information that
makes new instances recognisably members of the same class, such as
symmetries, grammars, constraints, conservation laws, dependency
structures or recurrence relations. For empirical systems, \(\sigma(\Class)\) may be specified by a benchmark or model-family description, and \(\MutInfo(\descop(\Hw):\sigma(\Class))\) may be estimated by resource-bounded MDL, compression gain, or benchmark-defined latent variables.

We can now define the algorithmic
counterparts of the three properties that characterise chemical
catalysts: pathway opening, non-consumption, and selectivity.
Pathway opening corresponds to a reduction in irreversible
bit-operations at matched intelligence, measured via
\eqref{eq:energy-floor}. Non-consumption requires that the
substrate returns close to a fixed reference state after each
cycle, with the associated restoration energy accounted for
explicitly. Selectivity becomes a condition on conditional
Kolmogorov complexity: the substrate must encode non-trivial
information about the task class, together with a transfer
requirement that rules out finite-instance memorisation. We formalise these in Definition~\ref{def:catalyst} below.  

\begin{definition}[Algorithmic speed-up and barrier]
\label{def:barrier}
For systems $\Sys_0$ and $\Sys_1$ on a task $\Task$ at matched
intelligence $\IntOp(\Sys_1,\Task)=\IntOp(\Sys_0,\Task)$, the
class-specific speed-up factor $\Gamma(\Task)$ and the logarithmic operation barrier $\Barrier(\Sys,\Task)$
are
\begin{equation}
\Gamma(\Task):=\frac{\Nop(\Sys_0,\Task)}{\Nop(\Sys_1,\Task)},
\qquad
\Barrier(\Sys,\Task):=\log_2\Nop(\Sys,\Task).
\end{equation}
\end{definition}
Definition~\ref{def:barrier} introduces two complementary quantities:
$\Gamma$ measures how much less work one system requires than another
at the same task (hence matched intelligence), while $\Barrier=\log_2\Nop$ expresses
this cost on a logarithmic scale, so that reductions in $\Barrier$
capture multiplicative improvements as the opening of lower-cost
computational pathways. With this notion of a barrier in place, we can define the algorithmic analogue of catalysis: a reusable computational structure that lowers
the effective computational barrier across a task class, thereby
opening pathways that would otherwise be infeasible under the same
resource constraints, while remaining available for repeated use.

\begin{definition}[Algorithmic catalyst]
\label{def:catalyst}
Let $\Sys_0=(\Alg_0,\Hw_{\mathrm{exec},0},\Hw_{\mathrm{adapt},0})$
be a reference system for a task class $\Class$, with irreversible
operation count $\Nop_0$ and intelligence score $\IntOp_0$. A system
$\Sys_{\mathrm{cat}}=(\Alg_{\mathrm{cat}},
\Hw_{\mathrm{exec,cat}},\Hw_{\mathrm{adapt,cat}})$ is an
\emph{algorithmic catalyst} for $\Sys_0$ on $\Class$ if the following
three conditions hold on every $\Task\subseteq\Class$.
\begin{enumerate}
\item \textbf{Pathway opening.} At matched intelligence
$\IntOp(\Sys_{\mathrm{cat}},\Task)=\IntOp_0$, deployment
irreversibility is strictly reduced,
$\Nop(\Sys_{\mathrm{cat}},\Task)<\Nop(\Sys_0,\Task)$ (where
deployment work excludes training, adaptation and restoration
erasures, which are accounted separately via $\Eadapt,\Erestore$ below).
\item \textbf{Bounded reconfiguration.} There exists
$\Delta_K^{\mathrm{cycle}}\in\mathbb{Z}_{\ge 0}$ such that after
each benchmark cycle the execution substrate lies within
$\Delta_K^{\mathrm{cycle}}$ bits of the idle description
$\idle(\Alg_{\mathrm{cat}})$ in prefix description length, and the
per-cycle restoration energy accounted on $\Hw_{\mathrm{adapt,cat}}$
satisfies
\begin{equation}
\label{eq:restore-energy}
\Erestore\;\ge\;
\Fover(\Hw_{\mathrm{adapt,cat}})\,\Delta_K^{\mathrm{cycle}}\,\Erasure.
\end{equation}
\item \textbf{Structural selectivity.} There exist $\delta>0$ and
$\eta>0$ such that the speed-up persists over arbitrarily large
finite subsets of the task class:\begin{equation}
\label{eq:transfer}
\liminf_{|\Task'|\to\infty,\,\Task'\subseteq\Class}
\Gamma(\Task')\;\ge\;1+\delta,
\end{equation}
and the substrate description conditionally compresses the class
descriptor by at least $\eta$ bits,
\begin{equation}
\label{eq:selectivity}
K_U\!\left(\sigma(\Class)\mid
\descop(\Hw_{\mathrm{exec,cat}})\right)\;\le\;
K_U(\sigma(\Class))-\eta.
\end{equation}
\end{enumerate}
The quantity $\eta$ is the \emph{structural information} carried by
the substrate about the task class.
\end{definition}
The pathway-opening condition is written in irreversible-bit units so
that reductions in $\Nop$ directly correspond to reductions in
physical cost via \eqref{eq:energy-floor} at matched
execution-substrate overhead $\frac{\WPI^{\downarrow}(\Sys_{\mathrm{cat}})}
     {\WPI^{\downarrow}(\Sys_0)} = \frac{1}{\Gamma}$.
A catalytic system is therefore one that makes previously costly
computational routes usable at lower energy, in the same sense that
a chemical catalyst makes an otherwise inaccessible reaction pathway
available. The bounded-reconfiguration
condition ensures that this advantage is reusable rather than
one-off: as in chemistry, the catalyst may change during the process
but must return close to a fixed reference state after each cycle,
with the work required to restore that state explicitly accounted
for rather than ignored. The structural-selectivity condition
captures the analogue of active-site specificity: the substrate must
encode information about the structure of the task class itself, so
that the reduction in cost persists on new instances drawn from that
class, rather than being limited to a finite set of memorised inputs.

A useful consequence of the framework is a natural notion of
\emph{refinement} between substrates, which we leverage when considering how algorithmic catalysts may be composed (Section~\ref{sec:composition}). Intuitively,
a substrate $\Hw'$ refines $\Hw$ if it contains all of the
structural information that $\Hw$ provides about the task class,
possibly together with additional structure. This is
captured by the condition $K_U(\descop(\Hw)\mid\descop(\Hw')) = O(1)$ meaning that $\Hw$ can be reconstructed from $\Hw'$ up to a
constant description-length overhead and is formalised in the following lemma:

\begin{lemma}[Substrate monotonicity]
\label{lem:monotone}
If $\Hw'$ refines $\Hw$, then
$\MutInfo(\descop(\Hw'):\sigma(\Class))\ge
\MutInfo(\descop(\Hw):\sigma(\Class))-c_U$, and in particular
$\eta'\ge\eta-c_U$.
\end{lemma}
\begin{proof}
Conditioning on $\descop(\Hw')$ allows reconstruction of
$\descop(\Hw)$ with $O(1)$ overhead by the refinement hypothesis,
so $K_U(\sigma(\Class)\mid\descop(\Hw'))\le
K_U(\sigma(\Class)\mid\descop(\Hw))+c_U$ by the standard coding
argument for conditional prefix complexity. Subtracting both sides
from $K_U(\sigma(\Class))$ yields the mutual-information inequality,
and the $\eta$-inequality follows from \eqref{eq:selectivity}.
\end{proof}
Monotonicity is of use because it allows catalytic improvements to accumulate in a controlled way:
once a substrate has encoded useful structure, any refinement
inherits that structure and can build on it, rather than having to
recover it from scratch. We now turn to the central result of the framework,
which formalises how much improvement such structure can support in
terms of the information it encodes about the task class.

\section{The Structural Selectivity Theorem}
\label{sec:selectivity}

A substrate cannot exploit structure it does not contain. In the same
way that a chemical catalyst can only accelerate reactions whose
transition states are stabilised by its geometry, a computational
substrate can only reduce cost for those aspects of a task class that
are reflected in its own structure. Any transferable, class-specific
speed-up must therefore arise from information about that class,
whether encoded in program text, trained weights, circuit layout,
memory organisation or control dynamics. We formalise these conditions via the structural selectivity theorem below, which isolates the
class-specific component of speed-up in a universal-search model and
separates it from generic implementation improvements.

\begin{definition}[Universal-search cost model]
\label{def:search-model}
A \emph{universal-search cost model} for a task class $\Class$ is a
cost function on algorithms that, given access to a substrate
description $S$ as conditional input, assigns to each prefix program
$p$ a solver cost on instances of size $n$ proportional to
$2^{|p|}\cdot\poly(n)$, and declares $p$ admissible on $\Class$
when $p$, conditioned on $S$, reconstructs enough of $\sigma(\Class)$
to select a polynomial-time solver for every
$\Task\subseteq\Class$. 
\end{definition}
Under this definition, generic implementation improvements that do
not depend on $\sigma(\Class)$ are absorbed into the reference
baseline, so that $\Gamma(\Task)$ of Definition~\ref{def:barrier}
measures only the class-specific component. Universal-search models of this kind originate with
Levin~\cite{levin1984randomness} and provide the natural setting in
which the shortest program reconstructing the class structure
dominates search cost. In the catalytic case, the substrate description
\(S=\descop(\Hw_{\mathrm{exec,cat}})\) is treated as conditional
information for reconstructing \(\sigma(\Class)\).

\begin{theorem}[Structural selectivity]
\label{thm:selectivity}
Let $\Sys_{\mathrm{cat}}$ be an algorithmic catalyst for $\Sys_0$ on
$\Class$ and let $S:=\descop(\Hw_{\mathrm{exec,cat}})$. In any
universal-search cost model satisfying Definition~\ref{def:search-model},
the class-specific barrier reduction
$\Barrier(\Sys_0,\Task)-\Barrier(\Sys_{\mathrm{cat}},\Task)
= \log_2\Gamma(\Task)$ from Definition \ref{def:barrier} satisfies
\begin{equation}
\label{eq:selectivity-bound}
\log_2\Gamma(\Task)\;\le\;
\MutInfo(\descop(\Hw_{\mathrm{exec,cat}}):\sigma(\Class))+c_U
\;=\;K_U(\sigma(\Class))-K_U(\sigma(\Class)\mid S)+c_U.
\end{equation}
\end{theorem}
\begin{proof}
Let $p^*$ denote the shortest prefix program that reconstructs
$\sigma(\Class)$ to the accuracy required by
Definition~\ref{def:search-model} without supplementary information, so that
$|p^*|=K_U(\sigma(\Class))+O(1)$, and let $p^*_S$ denote the shortest
prefix program reconstructing $\sigma(\Class)$ with $S$ as conditional
input, so that $|p^*_S|=K_U(\sigma(\Class)\mid S)+O(1)$. By the
definition of the cost model, any admissible solver selection
reduces to emitting such a program, whose dominant multiplicative
factor is $2^{|p^*|}$ in the absence of side information and
$2^{|p^*_S|}$ when $S$ is available. The ratio of the two upper-bounds
the class-specific speed-up attributable to $S$:
\begin{equation}
\Gamma(\Task)\;\le\;2^{\,|p^*|-|p^*_S|+O(1)}
\;=\;2^{\,K_U(\sigma(\Class))-K_U(\sigma(\Class)\mid S)+O(1)}.
\end{equation}
Taking logarithms and absorbing the $O(1)$ term into the universal machine constant $c_U$ gives
\eqref{eq:selectivity-bound}. 
\end{proof}
Improvements that do not depend on the structure of the task class are
treated as part of the baseline, so that $\Gamma$ measures only the
reduction coming from exploiting $\sigma(\Class)$ itself. The transfer
condition \eqref{eq:transfer} then rules out systems that merely
memorise a finite set of instances. Such a system may reduce work on a
fixed benchmark, but as the class is enlarged it provides no
information about how new instances are generated. In that case the
substrate does not shorten the description of $\sigma(\Class)$,
so $K_U(\sigma(\Class)\mid S)=K_U(\sigma(\Class))+O(1)$, and the
apparent speed-up disappears in the limit.

Two immediate consequences of Theorem~\ref{thm:selectivity} clarify
what limits catalytic improvement. Let
\[
\mu :=
\MutInfo(\descop(\Hw_{\mathrm{exec,cat}}):\sigma(\Class))
\]
denote how much information the substrate actually carries about
the structure of the task class. The structural-information
parameter $\eta$ of Definition~\ref{def:catalyst} cannot exceed
this quantity, $\eta \le \mu$, so the achievable speed-up is
directly constrained by how much of the class structure is already
encoded in the substrate. In the regime where solving the class is
essentially equivalent to recovering $\sigma(\Class)$, this bound is
tight up to the universal constant $c_U$.

This interpretation makes clear why simple lookup mechanisms fail
to qualify as catalysts. As the corollary below shows, a cache (i.e. a finite table storing
precomputed input/output pairs) can reduce cost on the stored
instances, but it does not encode the underlying structure of the
task class. As the class is enlarged, its advantage does not
transfer, because it provides no information about how new
instances should be solved. By contrast, a genuine catalyst
embodies class-level structure and continues to reduce cost across
new instances drawn from the same class, just as a chemical
catalyst selectively accelerates an entire family of reactions
rather than a finite list of outcomes. 

\begin{corollary}[Cache non-example]
\label{cor:cache}
A finite lookup table $\Hw_{\mathrm{cache}}$ storing fixed
input-output pairs on a bounded set
$\Task_{\mathrm{stored}}\subseteq\Class$ satisfies
$\liminf_{|\Task'|\to\infty}\Gamma(\Task')=1$, and therefore cannot
be an algorithmic catalyst for $\Sys_0$ on $\Class$ in the sense of
Definition~\ref{def:catalyst}.
\end{corollary}
\begin{proof}
On any $\Task'$ with $|\Task'\setminus\Task_{\mathrm{stored}}|
\to\infty$, the proportion of instances covered by the cache
vanishes, so the overall work approaches that of $\Sys_0$, giving
$\Gamma(\Task')\to 1$. From the information perspective, the
description of $\Hw_{\mathrm{cache}}$ encodes only finitely many
instance-level answers rather than the class structure
$\sigma(\Class)$, so
\[
K_U(\sigma(\Class)\mid\descop(\Hw_{\mathrm{cache}}))
\ge
K_U(\sigma(\Class)) - O(1),
\]
and hence the mutual information with the class is bounded. By
Theorem~\ref{thm:selectivity}, this prevents any unbounded
class-level speed-up.
\end{proof}

\section{Thermodynamic--Informational Coupling}
\label{sec:coupling}

Theorem~\ref{thm:selectivity} is an information-theoretic statement:
it limits how much speed-up can be obtained from the structure a
catalyst encodes, but does not yet relate this to physical cost.
The link to energy comes from the fact that any such structure must
have been established in the substrate during adaptation, and
writing information into a physical system may require irreversible
operations unless that information is already supplied as input.
This is the content of the Zurek--Bennett correspondence between
algorithmic and thermodynamic entropy~\cite{zurek1989algorithmic,bennett1982thermodynamics,bennett1998information}: information
present in the system reflects work that was done to put it there.

To make this precise, let $\mathcal{D}$ denote the adaptation input,
understood broadly to include training data, design specifications,
source code, physical parameters, or any other information available
to the adaptation process without additional irreversible work on
$\Hw_{\mathrm{adapt,cat}}$. Let $S_0$ denote the initial state of the
execution substrate before adaptation, chosen so that
\[
K_U(\sigma(\Class)\mid S_0)=K_U(\sigma(\Class))+O(1),
\]
where $S_0$ contains no useful information about the class
structure. This ensures that any structural information present
after adaptation must have been introduced during the adaptation
process itself. Finally, let $\Herase$ denote the number of logical
erasures performed on $\Hw_{\mathrm{adapt,cat}}$ during adaptation. The remaining question is how much physical work is required to
establish this structure during adaptation. Any information present
in the execution substrate after adaptation must come either from
the input $\mathcal{D}$ or from irreversible operations performed
during the process. The following lemma makes this constraint
explicit.

\begin{lemma}[Physical erasures lower-bound installed information]
\label{lem:zurek-bennett}
Any physical adaptation process on $\Hw_{\mathrm{adapt,cat}}$ that
takes the pair $(S_0,\mathcal{D})$ to a state in which the
execution-substrate component has algorithmic mutual information
$\mu$ with $\sigma(\Class)$ satisfies
\begin{equation}
\label{eq:zurek-bennett}
\Herase\;\ge\;\mu\,-\,\MutInfo(\mathcal{D}:\sigma(\Class))\,-\,c_U.
\end{equation}
\end{lemma}
\begin{proof}
Any physical computation with $\Herase$ logical erasures admits,
by Bennett's reversible simulation
theorem \cite{bennett1973logical,bennett1989spacetime}, a reversible simulation using an auxiliary advice string
\(\mathrm{adv}\), with \(|\mathrm{adv}|\le \Herase+O(1)\), recording
the information lost in those erasures. Let
$\tilde{\mathcal{A}}$ denote such a reversible simulation of the
adaptation process, taking the augmented input
$(S_0,\mathcal{D},\mathrm{adv})$ with $|\mathrm{adv}|\le\Herase+O(1)$
to the augmented output
$(S,\mathcal{D}',\mathrm{hist})$, where $S$ is the post-adaptation
execution-substrate state and $\mathrm{hist}$ records the evolution.
Because $\tilde{\mathcal{A}}$ is reversible and has constant-size
description, the data-processing inequality for algorithmic mutual
information~\cite{li2019kolmogorov,hutter2004universal} gives
\begin{equation}
\label{eq:data-processing}
\MutInfo(S:\sigma(\Class))\;\le\;
\MutInfo\bigl((S_0,\mathcal{D},\mathrm{adv}):\sigma(\Class)\bigr)+O(\log),
\end{equation}
since any program computing $\sigma(\Class)$ from $S$ can be
composed with the inverse of $\tilde{\mathcal{A}}$ applied to the
full output tuple at a cost of $O(1)$ program bits. Applying the
Kolmogorov chain rule for algorithmic mutual
information~\cite{li2019kolmogorov} to the triple on the right:
\begin{align}
\MutInfo\bigl((S_0,\mathcal{D},\mathrm{adv}):\sigma(\Class)\bigr)
\;&\le\;
\MutInfo(S_0:\sigma(\Class))\notag\\
&+\MutInfo(\mathcal{D}:\sigma(\Class)\mid S_0)
+\MutInfo(\mathrm{adv}:\sigma(\Class)\mid S_0,\mathcal{D})+O(\log).
\label{eq:chainrule}
\end{align}
The first term is $O(1)$ by the choice of $S_0$; since conditioning
on a string carrying only $O(1)$ information about the target
changes mutual information by at most $O(\log)$, the second term
satisfies $\MutInfo(\mathcal{D}:\sigma(\Class)\mid S_0)\le
\MutInfo(\mathcal{D}:\sigma(\Class))+O(\log)$; and the third term is
bounded above by $|\mathrm{adv}|+O(1)\le\Herase+O(1)$ because
algorithmic mutual information is bounded above by the length of
either argument. Collecting logarithmic additive terms into $c_U$,
\begin{equation}
\mu\;=\;\MutInfo(S:\sigma(\Class))\;\le\;
\MutInfo(\mathcal{D}:\sigma(\Class))+\Herase+c_U,
\end{equation}
which rearranges to \eqref{eq:zurek-bennett}.
\end{proof}
Lemma~\ref{lem:zurek-bennett} provides the link between
information and physical cost. It says that if the execution
substrate ends up encoding $\mu$ bits of information about the
task class, that information must have come from somewhere: either
it was already present in the adaptation input $\mathcal{D}$, or it
was created during adaptation through irreversible operations. In
the latter case, the information must be recorded — via the advice
bits in the reversible simulation — and each such bit corresponds to
at least one logical erasure. Since erasures have a minimum energy
cost of $\Erasure$ per bit on the adaptation substrate, this places
a direct lower bound on the physical work required to encode
structural information. The following Theorem~\ref{thm:adaptation-cost} makes this
explicit by translating the erasure bound into an energy bound:
only the information not already supplied by $\mathcal{D}$ must be
paid for thermodynamically, and the cost scales linearly with that
residual information.

\begin{theorem}[Thermodynamic cost of structural information]
\label{thm:adaptation-cost}
The adaptation energy $\Eadapt$ on $\Hw_{\mathrm{adapt,cat}}$ of any system
$\Sys_{\mathrm{cat}}$ with post-adaptation substrate mutual
information $\mu$ and adaptation input $\mathcal{D}$ satisfies
\begin{equation}
\label{eq:adaptation-floor}
\Eadapt\;\ge\;
\Fover(\Hw_{\mathrm{adapt,cat}})\,\Erasure\,
\bigl[\mu\,-\,\MutInfo(\mathcal{D}:\sigma(\Class))\,-\,c_U\bigr]_+,
\end{equation}
where $[x]_+:=\max(x,0)$.
\end{theorem}
\begin{proof}
Combine Lemma~\ref{lem:zurek-bennett}, which shows that at least
$\Herase$ logical erasures are required to encode the necessary
structural information, with the energy lower bound
\eqref{eq:energy-floor} applied to $\Hw_{\mathrm{adapt,cat}}$, which
implies that each such erasure incurs a minimum cost of
$\Fover(\Hw_{\mathrm{adapt,cat}})\Erasure$. This yields the lower bound on $\Eadapt$. The $[\cdot]_+$ requirement enforces
non-negativity: if the adaptation input $\mathcal{D}$ already
provides as much (or more) structural information as the catalyst
ultimately contains, then no additional work is required to encode
it.
\end{proof}
Theorem~\ref{thm:adaptation-cost} supplies the thermodynamic half of
the coupling. The informational half is
Theorem~\ref{thm:selectivity}, which relates $\mu$ to the logarithmic
speed-up. Combining the two yields our central result regarding constraints upon algorithmic catalysis.

\begin{theorem}[Thermodynamic--informational coupling]
\label{thm:coupling}
Let $\Sys_{\mathrm{cat}}$ be an algorithmic catalyst for $\Sys_0$ on
$\Class$, with class-specific speed-up $\Gamma$, per-cycle
restoration energy $\Erestore$, adaptation input $\mathcal{D}$ and
adaptation energy $\Eadapt$ on $\Hw_{\mathrm{adapt,cat}}$. Let \(E_{0\to 1}\) denote the energy required for \(\Sys_0\) to
complete one deployment query and assume a matched execution-substrate overhead, so that
the catalytic per-query deployment energy is $E_{0\to 1}/\Gamma$.
Then
\begin{equation}
\label{eq:coupling-adapt}
\Eadapt\;\ge\;
\Fover(\Hw_{\mathrm{adapt,cat}})\,\Erasure\,
\bigl[\log_2\Gamma\,-\,\MutInfo(\mathcal{D}:\sigma(\Class))
\,-\,2c_U\bigr]_+,
\end{equation}
and, writing \(N_{\mathrm{inf}}^{*}\) for the minimum number of
deployment queries required to amortise adaptation, 
\begin{equation}
\label{eq:coupling-breakeven}
N_{\mathrm{inf}}^{*}\;\ge\;
\frac{\Fover(\Hw_{\mathrm{adapt,cat}})\,\Erasure\,
\bigl[\log_2\Gamma-\MutInfo(\mathcal{D}:\sigma(\Class))-2c_U\bigr]_+}
{E_{0\to 1}(1-1/\Gamma)-\Erestore},
\end{equation}
provided the denominator in \eqref{eq:coupling-breakeven} is
positive. Otherwise no break-even horizon exists.
\end{theorem}
\begin{proof}
Theorem~\ref{thm:selectivity} gives $\log_2\Gamma\le\mu+c_U$, so
$\mu\ge\log_2\Gamma-c_U$. Substituting this lower bound into
\eqref{eq:adaptation-floor} of Theorem~\ref{thm:adaptation-cost}
and absorbing the two constants into $2c_U$ yields
\eqref{eq:coupling-adapt}. The break-even count is obtained by
setting the amortised catalytic per-query energy
$E_{0\to 1}/\Gamma+\Erestore+\Eadapt/N_{\mathrm{inf}}$ equal to the
baseline per-query energy $E_{0\to 1}$, solving for
$N_{\mathrm{inf}}$, and substituting the lower bound
\eqref{eq:coupling-adapt} for $\Eadapt$. The denominator positivity
condition expresses that no positive deployment count can amortise
a fixed upfront cost when deployment itself does not save energy.
\end{proof}

Theorem~\ref{thm:coupling} makes explicit the central constraint of
the watts-per-intelligence framework: the speed-up delivered by a
catalyst and the cost of constructing it are tied by the same
underlying structure. Theorem~\ref{thm:selectivity} bounds how much
speed-up is possible from the information the substrate carries
about the task class, while Theorem~\ref{thm:adaptation-cost}
bounds the physical work required to encode that information.
Taken together, they show that achieving a logarithmic speed-up
$\log_2\Gamma$ requires introducing a comparable amount of
structural information into the system, unless that information is
already supplied by the adaptation input. In physical terms, this
information must be written into the substrate, and that process
has a minimum energy cost which must be amortised over deployment.
A catalyst that promises a speed-up $\Gamma$ without access to
sufficient class-level information in $\mathcal{D}$ cannot avoid a
proportional adaptation cost, and therefore cannot improve
watts-per-intelligence on short deployment horizons. 

Note also that, $E$, $\Eadapt$, $\Erestore$, and $E_{0\to 1}$ denote heat dissipated by logically irreversible operations, not the total energy drawn by the hardware. For reusable memory or substrate states, these quantities include the cost of preparing the state, erasing work bits, and returning the state for reuse, except where the relevant class information is already supplied in \(\mathcal{D}\). It is also further assumed that \(S_0\) contains no class information, so \(K_U(\sigma(\Class)\mid S_0)=K_U(\sigma(\Class))+O(1)\).

\section{Composition}
\label{sec:composition}
Our algorithmic catalysis framework is intended to explore the construction of complex
systems from simpler catalytic components. In chemistry, catalytic
pathways are composed by chaining reactions: rate improvements
multiply, while the underlying structural constraints interact
through shared transition states and intermediates. Our algorithmic
analogue follows the same pattern. Catalytic improvements combine
across stages, but not independently: each stage inherits and
extends the structure encoded by the previous ones, and the total
cost must account for both accumulated structure and repeated
restoration. The following theorem formalises how such catalytic
systems compose within the watts-per-intelligence framework. This suggests that catalytic improvements should combine in a
structured way: successive stages can build on previously installed
structure, leading to multiplicative gains in efficiency and
controlled accumulation of class-specific information. 

\begin{theorem}[Composition of catalysts]
\label{thm:composition}
Let $\Sys_0\to\Sys_1\to\Sys_2$ be a chain of algorithmic catalysts
on a common task class $\Class$, with stagewise class-specific
speed-ups $\Gamma_1,\Gamma_2$ and structural-information parameters
$\eta_1,\eta_2$, and assume that $\descop(\Hw_{\mathrm{exec},2})$
refines $\descop(\Hw_{\mathrm{exec},1})$ in the sense of
Lemma~\ref{lem:monotone}. Then the composite system
$\Sys_0\to\Sys_2$ satisfies
\begin{equation}
\label{eq:composition}
\Gamma_{1\circ 2}\;\ge\;\Gamma_1\cdot\Gamma_2,
\qquad
\eta_{1\circ 2}\;\ge\;\max\{\eta_1,\eta_2\}-c_U,
\end{equation}
with the stronger additive bound
$\eta_{1\circ 2}\ge\eta_1+\eta_2-c_U$ holding when the two stages
encode algorithmically independent aspects of $\sigma(\Class)$, in
the sense that $\MutInfo(\descop(\Hw_{\mathrm{exec},1}):
\descop(\Hw_{\mathrm{exec},2})\mid\sigma(\Class))=O(1)$. The
adaptation energy lower bound for the composite is
$\Eadapt^{1\circ 2}\ge\Fover(\Hw_{\mathrm{adapt},2})\Erasure
[\eta_{1\circ 2}-\MutInfo(\mathcal{D}_{\mathrm{tot}}:\sigma(\Class))
-c_U]_+$, where $\mathcal{D}_{\mathrm{tot}}$ is the concatenated
adaptation input of the two stages.
\end{theorem}
\begin{proof}
Multiplicativity of $\Gamma$ follows from the definitional
identity
$\Nop(\Sys_0,\Task)/\Nop(\Sys_2,\Task)
=(\Nop(\Sys_0,\Task)/\Nop(\Sys_1,\Task))\cdot
(\Nop(\Sys_1,\Task)/\Nop(\Sys_2,\Task))\ge\Gamma_1\Gamma_2$ at
matched intelligence. The refinement hypothesis and
Lemma~\ref{lem:monotone} yield $\eta_{1\circ 2}\ge\eta_1-c_U$ and
$\eta_{1\circ 2}\ge\eta_2-c_U$, hence
$\eta_{1\circ 2}\ge\max\{\eta_1,\eta_2\}-c_U$. For the independent
case, algorithmic independence of
$\descop(\Hw_{\mathrm{exec},1})$ and $\descop(\Hw_{\mathrm{exec},2})$
conditional on $\sigma(\Class)$ implies, up to $c_U$, that
$\MutInfo(\descop(\Hw_{\mathrm{exec},2}):\sigma(\Class))\ge
\MutInfo(\descop(\Hw_{\mathrm{exec},1}):\sigma(\Class))
+\MutInfo(\descop(\Hw_{\mathrm{exec},2}):\sigma(\Class)\mid
\descop(\Hw_{\mathrm{exec},1}))$ by the algorithmic chain rule,
and each term on the right lower-bounds the corresponding
$\eta_i$ up to $c_U$. The adaptation-energy bound follows from
Theorem~\ref{thm:adaptation-cost} applied to the composite.
\end{proof}

The theorem shows that successive catalytic stages can combine to
produce a large overall reduction in the effective barrier: each
stage builds on the structure established by earlier ones, so that
overlapping structure need not be reintroduced. In this case, the
gain is multiplicative but the accumulated structural information is
limited to what is not already shared across stages. By contrast,
when two stages encode different, non-overlapping aspects of the
task class, their contributions accumulate additively. This mirrors
the chemical case in which distinct enzymes act on different parts
of a pathway, each contributing its own selectivity without necessarily interfering with the others.

\section{Examples}
\label{sec:examples}
\subsection{Algorithmic catalysis on an affine-SAT class}
\label{sec:affine-sat}

All three conditions of Definition~\ref{def:catalyst} and the
coupling bound of Theorem~\ref{thm:coupling} can be verified in
closed form on a parametrised family of Boolean satisfiability
problems. This example is illustrative because it isolates the role of
structure: the task class is defined not by a finite list of
instances, but by a shared algebraic constraint that determines all
solutions. Let $V\subseteq\{0,1\}^n$ be an affine subspace of dimension $d$, and let $\Class_{n,d}$ be the class of $3$-SAT formulas $\varphi$ on
$n$ variables whose satisfying assignments are exactly the points
of $V$. Equivalently, each instance in the class encodes a different
presentation of the same underlying solution structure, namely the
subspace $V$. The canonical class descriptor
$\sigma(\Class_{n,d})$ therefore captures this shared structure by
specifying $V$ as an affine subspace. Concretely, this can be
achieved by giving a basis of $d$ vectors in $\{0,1\}^n$ together
with a constant offset vector, from which all satisfying assignments
can be generated. This representation makes explicit that the class
is defined by linear constraints rather than by individual
assignments, and yields
\[
K_U(\sigma(\Class_{n,d}))\le n d+n+O(\log n)
\]
bits up to standard encoding overhead.

A reference solver $\Sys_0$ for $\Class_{n,d}$ performs exhaustive
search over $\{0,1\}^n$, testing all possible assignments, and
therefore requires $\Nop_0=\Theta(2^n\cdot n)$ irreversible
operations per instance. This cost arises from treating each
instance independently, without exploiting the shared subspace
structure encoded by $\sigma(\Class_{n,d})$. A catalyst-substrate $\Hw_{\mathrm{exec,cat}}$ that encodes the
basis and offset of $V$ makes this structure explicit. Instead of
searching the full space $\{0,1\}^n$, the solver can enumerate only
the points in $V$, reducing the search to a space of size $2^d$.
This yields $\Nop_{\mathrm{cat}}=\Theta(2^d\cdot n)$ operations per
instance and a class-specific speed-up
\begin{equation}
\log_2\Gamma\;=\;n-d-O(\log n).
\end{equation}
The reduction reflects the fact that knowing $\sigma(\Class_{n,d})$
eliminates the need to explore directions orthogonal to $V$.

Pathway opening follows directly from the strict reduction in
$\Nop$ for $d<n$; matched intelligence holds because both solvers
return the same satisfying-assignment set with probability one. Bounded
reconfiguration follows because the catalytic substrate is used
without modification during deployment: it is only read, not
updated, so it remains close to its reference state across cycles.
As a result, $\Delta_K^{\mathrm{cycle}}=O(1)$ and the associated
restoration cost satisfies $\Erestore=O(\Erasure)$.

Structural selectivity can be verified explicitly in this setting.
The substrate description $\descop(\Hw_{\mathrm{exec,cat}})$
contains the basis and offset of $V$, and therefore already
encodes essentially all of the structure captured by
$\sigma(\Class_{n,d})$. As a result,
\[
K_U(\sigma(\Class_{n,d})\mid\descop(\Hw_{\mathrm{exec,cat}}))
=O(\log n),
\]
since only a small amount of additional information is needed to
reconstruct the full class descriptor from the substrate. In this representation, this corresponds to
\[
\mu=n d+n-O(\log n),
\]
and hence the structural-information parameter satisfies
$\eta=n d+n-O(\log n)$.

The transfer condition holds because the speed-up arises from the
shared subspace structure: for every
$\Task'\subseteq\Class_{n,d}$, the solver continues to operate on
$V$, giving
\[
\Gamma(\Task')=2^{n-d-O(\log n)}>1.
\]
In other words, the advantage does not depend on particular
instances but on the underlying structure of the class. Substituting into Theorem~\ref{thm:selectivity},
\begin{equation}
\log_2\Gamma\;=\;n-d-O(\log n)\;\le\;\mu+c_U\;=\;n d+n-O(\log n),
\end{equation}
so the selectivity bound holds with substantial slack whenever
$d\ge 1$. Intuitively, the substrate contains more information about
the subspace $V$ than is strictly needed to achieve the observed
speed-up, since eliminating the search directions orthogonal to $V$
requires only $n-d$ bits, whereas the full description of $V$
requires $n d + n$ bits.
The adaptation cost on this class makes the content of
Theorem~\ref{thm:adaptation-cost} explicit. An adaptation input
consisting of $m$ uniformly random satisfying assignments from $V$
can be viewed as revealing partial information about the underlying
subspace: each additional assignment constrains the space further. The resulting algorithmic mutual information
$\MutInfo(\mathcal{D}:\sigma(\Class_{n,d}))$ grows approximately
linearly in $m$ for $m\le d+1$ and saturates at $nd+n-O(\log n)$ bits
once $m\ge d+1$, since $d+1$ generic affine points determine $V$
uniquely. In other words, the training data progressively uncovers
the structure of $V$, until it becomes fully determined. The adaptation-energy lower bound from
Theorem~\ref{thm:adaptation-cost} is therefore
\begin{equation}
\label{eq:affine-adapt}
\Eadapt\;\ge\;
\Fover(\Hw_{\mathrm{adapt,cat}})\,\Erasure\,
\bigl[\,n d+n-\MutInfo(\mathcal{D}:\sigma(\Class_{n,d}))
-c_U\,\bigr]_+,
\end{equation}
which vanishes once \(m\gtrsim d+1\), because the relevant class information is then already supplied in the adaptation input. For $m<d+1$ training
assignments, the adaptation input carries only a partial description
of $\sigma(\Class_{n,d})$, and the remaining structure must be
supplied through irreversible operations on
$\Hw_{\mathrm{adapt,cat}}$: a training set of $m=d/2$ random
assignments at $n=100,d=10$ leaves a residual of approximately
$n(d-m+1)+n-O(\log n)\approx 600$ bits, giving
$\Eadapt\gtrsim 600\,\Fover(\Hw_{\mathrm{adapt,cat}})\Erasure$. Indicatively, at
room temperature $\Erasure\approx 2.87\times 10^{-21}\,\mathrm{J}$,
and with a contemporary overhead $\Fover\approx 10^{9}$ characteristic
of CMOS at the transistor level, the adaptation energy must satisfy
$\Eadapt\gtrsim 2\times 10^{-9}\,\mathrm{J}$.

The baseline per-query cost of exhaustive search is
$E_{0\to 1}\sim 2^n\cdot F\cdot\Erasure\approx 3.6\times 10^{18}\,
\mathrm{J}$ (up to polynomial factors and at the same overhead), reflecting the need to explore the
entire space $\{0,1\}^n$ without structural guidance. This is far
beyond any physically realistic energy budget, and illustrates why a
catalyst is required to make the class computationally accessible at
all. The catalytic per-query cost is
$E_{0\to 1}/\Gamma\sim 2^d\cdot F\cdot\Erasure\approx 3\times
10^{-9}\,\mathrm{J}$, roughly matching the adaptation cost, and the
break-even horizon of Theorem~\ref{thm:coupling} is well below a
single deployment query.

\subsection{Catalytic computation as a limit}
Finally, we show how the Buhrman--Cleve--Kouck{\'y}--Loff--Speelman model of catalytic computation \cite{buhrman2014computing} is recovered
in a precise limit.

\begin{proposition}[Zero-reconfiguration, structure-free limit]
\label{prop:cc-limit}
In the regime $\Delta_K^{\mathrm{cycle}}=0$, $\eta=0$,
$\Hw_{\mathrm{adapt,cat}}=\Hw_{\mathrm{exec,cat}}$, and
$\sigma(\Class)$ taken as a constant-length descriptor,
Definition~\ref{def:catalyst} reduces to a decision-problem instance
of the catalytic computation model: reusable auxiliary state with
exact restoration, no structural assumption on its contents, and
no thermodynamic cost beyond the Landauer floor. In this regime
Theorem~\ref{thm:coupling} degenerates to the trivial bound
$N_{\mathrm{inf}}^{*}\ge 0$.
\end{proposition}
\begin{proof}
The result follows by inspecting each condition in
Definition~\ref{def:catalyst} under the stated limits. Setting
$\Delta_K^{\mathrm{cycle}}=0$ enforces exact restoration, so the
execution substrate must return to its initial state after each
cycle, recovering the catalytic-tape condition of Buhrman
et~al.~\cite{buhrman2014computing}. Setting $\eta=0$ removes any
requirement that the substrate encode structure of the task class,
so the auxiliary state may be arbitrary. Taking
$\sigma(\Class)$ to have constant length ensures that both
$K_U(\sigma(\Class))$ and $K_U(\sigma(\Class)\mid S)$ are $O(1)$,
so the selectivity bound \eqref{eq:selectivity-bound} yields
$\log_2\Gamma\le c_U$, which does not constrain the speed-up at the
level of asymptotic complexity. Applying the universal-search cost
model of Definition~\ref{def:search-model} under these parameter
choices recovers the standard catalytic-computation setting with
speed-up not constrained by selectivity. Under these conditions, Definition~\ref{def:catalyst} reduces to
the standard catalytic computation model: reusable auxiliary state
with exact restoration and no structural assumptions. Substituting
into Theorem~\ref{thm:coupling} then gives $\Eadapt\ge 0$,
reflecting that, in the absence of structural requirements, the
adaptation-cost bound of Theorem~\ref{thm:coupling} imposes no
non-trivial thermodynamic lower bound.
\end{proof}

\section{Conclusion}
\label{sec:conclusion}
In this paper, we have presented a framework for algorithmic
catalysis, synthesising catalytic computation with thermodynamic
approaches to algorithmic intelligence. Motivated by the role of
catalysis in chemistry - where structure renders otherwise
inaccessible reaction pathways feasible - and the need for more thermodynamically-efficient algorithmic models of intelligence, we have shown that
algorithmic analogues of catalysts arise as reusable computational
structures characterised by pathway opening, bounded
reconfiguration and structural selectivity. Within this framework, the speed-up a substrate can achieve is
limited by how much of the task class structure it captures, and
acquiring that structure requires physical work unless it is
already provided as input. Taken together, this means that making a
class of computations thermodynamically accessible requires first
encoding that structure through irreversible operations, and that
the cost of doing so must be amortised over deployment. These results contribute to ongoing work on the theoretical and applied thermodynamics of machine
intelligence. Future work may examine how effective forms of
algorithmic catalysis arise within frontier reasoning models, and
the extent to which explicitly designing such structures can improve
the thermodynamic efficiency of intelligent algorithms.

\bibliographystyle{splncs04}
\bibliography{refs2.bib}

\begin{thebibliography}{10}
\providecommand{\url}[1]{\texttt{#1}}
\providecommand{\urlprefix}{URL }
\providecommand{\doi}[1]{https://doi.org/#1}

\bibitem{athas1994low}
Athas, W.C., Svensson, L.J., Koller, J.G., Tzartzanis, N., Chou, E.Y.C.: Low-power digital systems based on adiabatic-switching principles. {IEEE} Transactions on {VLSI} Systems  \textbf{2}(4),  398--407 (1994)

\bibitem{benkovic2003enzyme}
Benkovic, S.J., Hammes-Schiffer, S.: A perspective on enzyme catalysis. Science  \textbf{301}(5637),  1196--1202 (2003)

\bibitem{bennett1973logical}
Bennett, C.H.: Logical reversibility of computation. {IBM} Journal of Research and Development  \textbf{17}(6),  525--532 (1973)

\bibitem{bennett1982thermodynamics}
Bennett, C.H.: The thermodynamics of computation---a review. International Journal of Theoretical Physics  \textbf{21}(12),  905--940 (1982)

\bibitem{bennett1998information}
Bennett, C.H.: Notes on the history of reversible computation. {IBM} Journal of Research and Development  \textbf{32}(1),  16--23 (1988)

\bibitem{bennett1989spacetime}
Bennett, C.H.: Time/space trade-offs for reversible computation. SIAM Journal on Computing  \textbf{18}(4),  766--776 (1989)

\bibitem{buhrman2014computing}
Buhrman, H., Cleve, R., Kouck{\'y}, M., Loff, B., Speelman, F.: Computing with a full memory: catalytic space. In: Proceedings of the 46th Annual ACM Symposium on Theory of Computing (STOC 2014). pp. 857--866. ACM (2014)

\bibitem{buhrman2018catalytic}
Buhrman, H., Kouck{\'y}, M., Loff, B., Speelman, F.: Catalytic space: non-determinism and hierarchy. Theory of Computing Systems  \textbf{62}(1),  116--135 (2018)

\bibitem{cook2024tree}
Cook, J., Mertz, I.: Tree evaluation is in space {$O(\log n \cdot \log \log n)$}. In: Proceedings of the 56th Annual ACM Symposium on Theory of Computing (STOC 2024). pp. 1268--1278. ACM (2024)

\bibitem{cook2025lower}
Cook, J., Mertz, I.: Tree evaluation is in space {$O(\log n \cdot \log \log n)$}. SIAM Journal on Computing  (2025), journal version of STOC 2024

\bibitem{ebtekar2021information}
Ebtekar, A.: Information dynamics and the arrow of time. arXiv:2109.09709  (2021)

\bibitem{ebtekar2025foundations}
Ebtekar, A., Hutter, M.: Foundations of algorithmic thermodynamics. Physical Review E  \textbf{111}(1),  014118 (2025)

\bibitem{frank2005introduction}
Frank, M.P.: Introduction to reversible computing: motivation, progress, and challenges. In: Proceedings of the 2nd ACM Conference on Computing Frontiers. pp. 385--390. ACM (2005)

\bibitem{garcia2010enzyme}
Garcia-Viloca, M., Gao, J., Karplus, M., Truhlar, D.G.: How enzymes work: analysis by modern rate theory and computer simulations. Science  \textbf{303}(5655),  186--195 (2004)

\bibitem{girard2013catalytic}
Girard, V., Kouck{\'y}, M., McKenzie, P.: Nonuniform catalytic space and the direct sum for space. Tech. Rep. TR15-138, Electronic Colloquium on Computational Complexity (ECCC) (2015), \url{https://eccc.weizmann.ac.il/report/2015/138}

\bibitem{goldreich2024lower}
Goldreich, O.: Solving tree evaluation in {$o(\log n \cdot \log \log n)$} space. Tech. Rep. TR24-124, Electronic Colloquium on Computational Complexity (ECCC) (2024), \url{https://eccc.weizmann.ac.il/report/2024/124}

\bibitem{goldt2017stochastic}
Goldt, S., Seifert, U.: Stochastic thermodynamics of learning. Physical Review Letters  \textbf{118}(1),  010601 (2017)

\bibitem{grunwald2007minimum}
Gr{\"u}nwald, P.D.: The Minimum Description Length Principle. MIT Press, Cambridge, MA (2007)

\bibitem{hutter2004universal}
Hutter, M.: Universal Artificial Intelligence: Sequential Decisions Based on Algorithmic Probability. Springer, Berlin (2004)

\bibitem{kingma2013autoencoding}
Kingma, D.P., Welling, M.: Auto-encoding variational {B}ayes. In: Proceedings of the 2nd International Conference on Learning Representations (ICLR 2014) (2014)

\bibitem{kolmogorov1965three}
Kolmogorov, A.N.: Three approaches to the quantitative definition of information. Problems of Information Transmission  \textbf{1}(1), ~1--7 (1965)

\bibitem{koucky2016catalytic}
Kouck{\'y}, M.: Catalytic computation. Bulletin of the {EATCS} (118) (2016)

\bibitem{laidler1987chemical}
Laidler, K.J.: Chemical Kinetics. Harper and Row, New York, 3 edn. (1987)

\bibitem{legg_universal_2007}
Legg, S., Hutter, M.: Universal intelligence: a definition of machine intelligence. Minds and Machines  \textbf{17}(4),  391--444 (2007)

\bibitem{levin1984randomness}
Levin, L.A.: Randomness conservation inequalities; information and independence in mathematical theories. Information and Control  \textbf{61}(1),  15--37 (1984)

\bibitem{li2019kolmogorov}
Li, M., Vit{\'a}nyi, P.M.B.: An Introduction to {K}olmogorov Complexity and Its Applications. Springer, 4th edn. (2019)

\bibitem{mcallester1999pacbayes}
McAllester, D.A.: {PAC}-{B}ayesian model averaging. Machine Learning  \textbf{37},  355--363 (1999)

\bibitem{parrondo2015thermodynamics}
Parrondo, J.M.R., Horowitz, J.M., Sagawa, T.: Thermodynamics of information. Nature Physics  \textbf{11}(2),  131--139 (2015)

\bibitem{perrier2025watts}
Perrier, E.: Watts-per-intelligence: {Part I} (energy efficiency). In: International Conference on Artificial General Intelligence. Lecture Notes in Computer Science, vol. 16058, pp. 46--57. Springer (2025)

\bibitem{pyne2023catalytic}
Pyne, E.: Derandomizing logspace with a small shared hard drive. In: Proceedings of the 39th Computational Complexity Conference (CCC 2024). LIPIcs, vol.~300, pp. 4:1--4:20. Schloss Dagstuhl (2024), preliminary version: ECCC TR23-168 (2023)

\bibitem{rezende2014stochastic}
Rezende, D.J., Mohamed, S., Wierstra, D.: Stochastic backpropagation and approximate inference in deep generative models. In: Proceedings of the 31st International Conference on Machine Learning (ICML 2014). pp. 1278--1286 (2014)

\bibitem{rissanen1978modeling}
Rissanen, J.: Modeling by shortest data description. Automatica  \textbf{14}(5),  465--471 (1978)

\bibitem{solomonoff1964formal}
Solomonoff, R.J.: A formal theory of inductive inference. Information and Control  \textbf{7}(1--2),  1--22, 224--254 (1964)

\bibitem{still2012thermodynamics}
Still, S., Sivak, D.A., Bell, A.J., Crooks, G.E.: Thermodynamics of prediction. Physical Review Letters  \textbf{109}(12),  120604 (2012)

\bibitem{stuhlmuller2013amortized}
Stuhlm{\"u}ller, A., Taylor, J., Goodman, N.D.: Learning stochastic inverses. In: Advances in Neural Information Processing Systems 26 (NeurIPS 2013). pp. 3048--3056 (2013)

\bibitem{wolfenden2001depth}
Wolfenden, R., Snider, M.J.: The depth of chemical time and the power of enzymes as catalysts. Accounts of Chemical Research  \textbf{34}(12),  938--945 (2001)

\bibitem{wolpert2019stochastic}
Wolpert, D.H.: The stochastic thermodynamics of computation. Journal of Physics A: Mathematical and Theoretical  \textbf{52}(19),  193001 (2019)

\bibitem{zurek1989algorithmic}
Zurek, W.H.: Algorithmic randomness and physical entropy. Physical Review A  \textbf{40}(8),  4731--4751 (1989)

\end{thebibliography}

\end{document}